\documentclass[twocolumn]{IEEEtran}

\usepackage{amsmath}
\usepackage{amsthm,amsfonts,paralist,mathtools}
\usepackage[hidelinks]{hyperref}
\usepackage{pgfplots}
\usepackage{tikz}
\usepackage{authblk}
\usetikzlibrary{matrix}
\usetikzlibrary{decorations.pathreplacing}
\usetikzlibrary{quantikz}

\DeclareMathOperator{\Tr}{Tr}
\DeclareMathOperator{\wt}{wt}

\newtheorem{proposition}{Proposition}
\newtheorem{theorem}[proposition]{Theorem}
\newtheorem{lemma}[proposition]{Lemma}

\newtheorem{corollary}[proposition]{Corollary}
\newtheorem{conjecture}[proposition]{Conjecture}
\theoremstyle{definition}
\newtheorem{example}[proposition]{Example}

\pgfplotsset{every axis legend/.append style={at={(1,0)},anchor=south east}}

\begin{document}
\nocite{*}
\title{Quantum Data-Syndrome Codes:\\ Subsystem and Impure Code Constructions}

\author{Andrew Nemec
\thanks{A. Nemec is with the Duke Quantum Center, Duke University, Durham, NC 27701, USA and the Department of Electrical and Computer Engineering, Duke University, Durham, NC 27708, USA (email: andrew.nemec@duke.edu).}}

\maketitle

\begin{abstract}
Quantum error correction requires the use of error syndromes derived from measurements that may be unreliable. Recently, quantum data-syndrome (QDS) codes have been proposed as a possible approach to protect against both data and syndrome errors, in which a set of linearly dependent stabilizer measurements are performed to increase redundancy. Motivated by wanting to reduce the total number of measurements performed, we introduce QDS subsystem codes, and show that they can outperform similar QDS stabilizer codes derived from them. We also give a construction of single-error-correcting QDS stabilizer codes from impure stabilizer codes, and show that any such code must satisfy a variant of the quantum Hamming bound for QDS codes. Finally, we use this bound to prove a new bound that applies to impure, but not pure, stabilizer codes that may be of independent interest.

\end{abstract}

\section{Introduction}

Active quantum error correction plays an important role in the protection of quantum information from external and internal noise. Unforturnately, since the physical gates that perform the error-correction procedure are themselves imperfect, we must design these procedures so that they are fault tolerant. Multiple fault-tolerant error-correction procedures have been developed, including those by Shor \cite{DiVincenzo1996, Shor1996}, Steane \cite{Steane1997}, and more recently the flag fault-tolerant scheme of Chao and Reichardt \cite{Chao2018}. In such a system, the failure of individual components during the error-correction procedure can be corrected.

The most well studied class of quantum error-correcting codes are the stabilizer codes \cite{Calderbank1998, Gottesman1997}. The error-correction procedure on these codes involves measuring the generators of the stabilizer group of the code, resulting in a classical error syndrome that provides information on what errors may have occured on the code. Because the gates that perform measurements may themselves be faulty, the syndrome may be incorrect, leading to the incorrect recovery operation being applied to the code. While some codes such as the 7-qubit Steane code have the ability to protect against syndrome errors \cite{Fujiwara2014}, the standard theory of stabilizer codes does not provide a way to protect against these syndrome measurement errors.

In Shor-style syndrome extraction and its improvements \cite{Delfosse2020, Tansuwannont2022}, syndrome measurement errors are protected against by repeated measurement of each stabilizer generator. This effectively means that each classical bit in the syndrome is protected by it's own classical repetition code. While this extra redundancy is useful for protecting against other faulty-gate errors, it is inefficient when looking at syndrome measurement errors. In a series of papers, Ashikhmin et al. \cite{Ashikhmin2014, Ashikhmin2016, Ashikhmin2020} and Fujiwara \cite{Fujiwara2014} developed the theory of quantum data-syndrome (QDS) codes, in which extra redundancy is introduced more efficiently by means of measuring an over-determined set of stabilizer elements. The most efficient way to construct, encode, and decode QDS codes is to use syndrome-measurement (SM) codes, where the syndrome of a stabilizer code is encoded in a classical linear code.

More recent work on QDS codes has included more efficient decoding techniques \cite{Kuo2021, Raveendran2022}, connections to the theory of 2-designs \cite{Premakumar2021}, and an extension to quantum convolutional codes \cite{Zeng2019}. Recently, Wagner et al. \cite{Wagner2022} showed how QDS codes can be used to estimate a logical channel with a phenomenological Pauli noise model. Parallel to the development of QDS codes, a similar theory of single-shot error correction \cite{Bombin2015, Brown2016, Campbell2019} has been investigated for use on topological quantum codes.

This paper extends quantum data-syndrome codes from stabilizer codes to subsystem codes, showing how they can provided an advantage over gauge-fixed versions of themselves. We then give a new construction of QDS codes from impure distance-3 stabilizer codes, as well as show that any such code must satisfy the QDS variant of the Hamming bound. Finally, we give a new bound on impure stabilizer codes that bounds them away from the quantum Hamming bound.

\section{Background}

\subsection{Quantum Stabilizer Codes}


An $\left(\!\left(n,K\right)\!\right)$ quantum code is a $K$-dimensional subspace of the Hilbert space $\mathcal{H}=\mathbb{C}^{2^{n}}$ that represents $n$ physical qubits. One nice class of quantum code are the stabilizer codes \cite{Calderbank1998, Gottesman1997}, which may be defined using a stabilizer formalism. Let $X_{j}$ be the Pauli-$X$ operator acting on the $j$-th physical qubit, and similarly for $Z_{j}$. Then the Pauli error group $G_{n}$ on $n$ qubits is generated by \begin{equation} \left\langle i, Z_{1}, \dots,Z_{n}, X_{1},\dots, X_{n}\right\rangle. \end{equation} To define a quantum stabilizer code, we pick a set of $2n$ independent operators $\left\{Z_{1}',\dots, Z_{n}',X_{1}', \dots, X_{n}'\right\}$ from $G_{n}$ that satisfy the same commutation relations as the single-qubit Pauli operators:
\begin{align}
\left[X_{i}',X_{j}'\right] & = 0, i,j\in\left[n\right], \label{comrel1} \\
\left[Z_{i}',Z_{j}'\right] & = 0, i,j\in\left[n\right], \label{comrel2} \\
\left[X_{i}',Z_{j}'\right] & = 0, i,j\in\left[n\right], i\neq j, \label{comrel3} \\
\left\{X_{i}',Z_{i}'\right\} & = 0, i\in\left[n\right], \label{comrel4} 
\end{align}
where the commutator $\left[A,B\right]=AB-BA$ and the anticommutator $\left\{A,B\right\}=AB+BA$ each equal $0$ if and only if $A$ and $B$ commute or anticommute respectively. These operators, together with the imaginary unit $i$, also generate the Pauli error group $G_{n}$ and can be thought of as single-qubit Pauli operators that act on $n$ virtual qubits.

A stabilizer code $\mathcal{C}$ is defined by its stabilizer group $\mathcal{S}$, which is an abelian subgroup of $G_{n}$ generated by $n-k$ independent, commuting elements $S_{i}$ and does not contain scalar multiples of the identity. The coding subspace $\mathcal{C}$ is the joint $+1$-eigenspace of the operators in $\mathcal{S}$ of dimension $K=2^{k}$. The centralizer $C_{G_{n}}\!\!\left(\mathcal{S}\right)$ of $\mathcal{S}$ in $G_{n}$ is the set of elements in $G_{n}$ that commute with every element of $\mathcal{S}$. Following \cite{Ashikhmin2020}, we will use the notation $m=n-k$ for the number of measurements made by a standard stabilizer code in order to simplify expressions.

Without loss of generality, we can choose our stabilizer generators such that $S_{i}=Z_{i}'$, for all $i\in\left[m\right]$, which means that \begin{equation} C_{G_{n}}\!\!\left(\mathcal{S}\right) = \left\langle i, Z_{1}', \dots, Z_{n}', X_{m+1}', \dots, X_{n}'\right\rangle \end{equation} by the commutation relations in Eqs. \ref{comrel1}-\ref{comrel4}. The stabilizer code $\mathcal{C}$ is therefore the $k$ virtual qubits indexed by $\left\{m+1, \dots, n\right\}$ and has the logical Pauli operators $\left\{X_{j}', Z_{j}'\right\}_{j=m+1,\dots,n}$.

The weight $\wt\!\left(E\right)$ of a Pauli operator $E$ is the number of non-identity tensor components it has. We define the minimum distance $d$ of a stabilizer code as \begin{equation} d=\min\wt\!\left(C_{G_{n}}\!\!\left(\mathcal{S}\right)\setminus\left\langle i, \mathcal{S}\right\rangle\right). \end{equation} We call a stabilizer code of dimension $2^{k}$ on $n$ physical qubits and a minimum distance of $d$ an $\left[\!\left[n,k,d\right]\!\right]$ stabilizer code.

Binary stabilizer codes have a nice connection to additive codes over $\mathbb{F}_{4}=\left\{0,1,\omega,\omega^{2}=\omega+1\right\}$ \cite{Calderbank1998}, using the homomorphism $\tau:G_{n}\rightarrow\mathbb{F}_{4}^{n}$ by corresponding the tensor components $\left\{I,X,Z,Y\right\}$ with $\left\{0,1,\omega,\omega^{2}\right\}$ respectively. We define the trace inner product of two vectors $\mathbf{u},\mathbf{v}\in\mathbb{F}_{4}^{n}$ by \begin{equation}\label{traceinnerprod}\mathbf{u}*\mathbf{v}=\Tr_{\mathbb{F}_{2}}^{\mathbb{F}_{4}}\!\left(\sum_{i=1}^{n}u_{i}\overline{v_{i}}\right)=\sum_{i=1}^{n}\left(u_{i}\overline{v_{i}}+\overline{u_{i}}v_{i}\right), \end{equation} where $\overline{v_{i}}$ is conjugation in $\mathbb{F}_{4}$ defined by $\overline{0}=0$, $\overline{1}=1$, $\overline{\omega}=\omega^{2}$, and $\overline{\omega^{2}}=\omega$. Notice that the trace inner product captures the commutation relationship between any two elements of $E,F\in G_{n}$, which either commute or anticommute with each other: $E$ and $F$ commute with each other if $\tau\!\left(E\right)*\tau\!\left(F\right)=0$ and anticommute with each other if $\tau\!\left(E\right)*\tau\!\left(F\right)=1$. In this way we can associate the stabilizer group $\mathcal{S}$ with an additive $\left(n,2^{m}\right)_{4}$ code $C$ and its centralizer $C_{G_{n}}\!\!\left(\mathcal{S}\right)$ with the dual $\left(n,2^{n+k}\right)_{4}$ code $C^{\perp}$. Since $\mathcal{S}\leq C_{G_{n}}\!\!\left(\mathcal{S}\right)$, we have $C\subseteq C^{\perp}$. The minimum distance of the stabilizer code may therefore also be written in terms of $C$ and $C^{\perp}$: \begin{equation*} d=\min\wt\!\left(C^{\perp}\setminus C\right). \end{equation*}

Let $\mathbf{g}_{i}=\tau\!\left(S_{i}\right), i\in\left[m\right]$, and let \begin{equation*}G_{C}=\begin{bmatrix} \mathbf{g}_{1} \\ \vdots \\ \mathbf{g}_{m} \end{bmatrix}\end{equation*} be the generator matrix of $C$. Suppose an encoded vector is affected by a Pauli error $E\in G_{n}$, and denote the corresponding $\mathbb{F}_{4}$ error-vector as $\mathbf{e}=\tau\!\left(E\right)$. The syndrome associated with this error is $\mathbf{s}\in\mathbb{F}_{2}^{m}$, where each element of $\mathbf{s}$ is given by $s_{i}=\mathbf{g}_{i}*\mathbf{e}$. The syndrome $\mathbf{s}$ shows whether the error commutes or anticommutes with the stabilizer generators, and is obtained by measuring the stabilizer generators $S_{i}$.

\subsection{Subsystem Codes}

Here we review the basics of subsystem codes using the stabilizer formalism of Poulin \cite{Poulin2005}. Instead of encoding using the entire coding subspace, we impose a subsystem structure on $\mathcal{C}=A\otimes B$ and encode our quantum message only in the subsystem $A$, while we may encode an arbitrary state into the gauge subsystem $B$. Upon recovery, we only demand that subsystem $A$ is correctable \cite{Kribs2013}.

We define a gauge group $\mathcal{G}\subseteq G_{n}$ such that $\mathcal{G}$ is generated by the imaginary unit $i$, $m-r$ independent, commuting elements $S_{1},\dots, S_{m-r}$, and $r$ pairs of operators $G_{j},H_{j}$ that satisfy the following commutation relations: \begin{align}
\left[S_{i},G_{j}\right] & = 0, i\in\left[m-r\right], j\in\left[r\right], \\
\left[S_{i},H_{j}\right] & = 0, i\in\left[m-r\right], j\in\left[r\right], \\
\left[G_{i},H_{j}\right] & = 0, i\in\left[m-r\right], j\in\left[r\right], i\neq j, \\
\left\{G_{i},H_{i}\right\} & = 0, i\in\left[r\right].
\end{align}


As with stabilizer codes, we let $\mathcal{S}=\left\langle S_{1}, \dots, S_{m-r}\right\rangle$ be the stabilizer group of the subsystem code, and let $C_{G_{n}}\!\!\left(\mathcal{S}\right)$ be its centralizer in $G_{n}$. Without loss of generality, we can make the following associations: \begin{align}S_{i} & = Z_{i}', i\in\left[m-r\right], \\ G_{i} & = Z_{m-r+i}', i\in\left[r\right], \\ H_{i} & = X_{m-r+i}', i\in\left[r\right], \end{align} so \begin{equation} \mathcal{G} = \left\langle i, Z_{1}', \dots, Z_{m}', X_{m-r+1}', \dots, X_{m}'\right\rangle,  \end{equation} and \begin{equation} C_{G_{n}}\!\!\left(\mathcal{S}\right) = \left\langle i, Z_{1}', \dots, Z_{n}', X_{m-r+1}', \dots, X_{n}' \right\rangle.\end{equation} Note that this gives $\mathcal{S}\subseteq\mathcal{G}\subseteq C_{G_{n}}\!\!\left(\mathcal{S}\right)$, and we define the minimum distance of the subsystem code to be \begin{equation} d=\min\wt\!\left(C_{G_{n}}\!\!\left(\mathcal{S}\right)\setminus\mathcal{G}\right). \end{equation} The coding subsystem $A$ is then the $k$ virtual qubits indexed by $\left\{m+1, \dots,n\right\}$, and the guage subsystem in the $r$ virtual qubits indexed by $\left\{m-r+1, \dots, m\right\}$.

Similar to the stabilizer case, we can associate an $\left(n, 2^{m-r}\right)_{4}$ classical additive code $C$ with the stabilizer group $\mathcal{S}$, as well as associate its dual $\left(n, 2^{n+k+r}\right)_{4}$ code $C^{\perp}$ with $C_{G_{n}}\!\!\left(\mathcal{S}\right)$. In the same way, we now associate an $\left(n,2^{m+r}\right)_{4}$ classical additive code $D$ with the gauge group $\mathcal{G}$, which satisfies $D\cap D^{\perp}=C$. The minimum distance of the subsystem code can therefore also be given in terms of the associated classical codes: \begin{equation} d=\min\wt\!\left(C^{\perp}\setminus D\right). \end{equation}

\subsection{Quantum Data-Syndrome Codes}

Much of the notation used in this and the following sections are taken from Ashikhmin et al. \cite{Ashikhmin2020}.

There are multiple protocols for extracting the syndrome of a stabilizer code \cite{Chao2018, Steane1997}, and here we assume that Shor's syndrome extraction protocol \cite{DiVincenzo1996, Shor1996} is used. The benefit of a Shor-style syndrome extraction scheme is that it makes use of transversal gates, so a single error to one of the ancilla qubits will propagate to at most a single data qubit. This is accomplished for a weight $w$ stabilizer generator by using a $w$-qubit ancilla cat state and performing $w$ measurements.

While all the gates involved in syndrome extraction have the potential to be faulty, for simplicity we assume that they perform perfectly, and that the only errors that occur are bit-flip errors on the syndrome caused by faulty measurements with probability, with $p_{m}$ the probability of a single-qubit measurement error. In particular, this means that we do not have to worry about errors propagating to the data qubits. The probability of incorrectly measuring a stabilizer generator $S_{i}$ of weight $w_{i}$ is given by \begin{equation} p_{err}\!\left(S_{i}\right)=\sum\limits_{j\text{ odd}}\binom{w_{i}}{j}p_{m}^{i}\left(1-p_{m}\right)^{w_{i}-j}. \end{equation}

We want to design stabilizer codes that can correct errors to the syndrome $\mathbf{s}\in\mathbb{F}_{n}^{m}$ as well as data errors. Following \cite{Ashikhmin2020}, define an inner product for vectors $u,v\in\mathbb{F}_{4}^{n}\times\mathbb{F}_{2}^{a}$, \begin{equation}\label{starinnerprod} u\star v=\Tr_{\mathbb{F}_{2}}^{\mathbb{F}_{4}}\!\left(\sum\limits_{i=1}^{n}u_{i}\overline{v_{i}}\right)+\sum\limits_{j=1}^{a}u_{n+j}v_{n+j}. \end{equation} Let $G_{C}$ generate the classical code $C$ associated with a stabilizer code, and define the matrix \begin{equation} \hat{G} = \begin{bmatrix} G_{C} & I_{m}\end{bmatrix}, \end{equation} which generates a code $\hat{C}\in\mathbb{F}_{4}^{n}\times\mathbb{F}_{2}^{m}$. By decoding the dual code $\hat{C}^{\perp}$ with regards to the $\star$-inner product in Eq. \ref{starinnerprod}, we can correct both data and syndrome errors for the stabilizer code associated with the code $C$. The minimum distance of the code is defined as \begin{equation} d=\min\wt\!\left(\hat{C}^{\perp}\setminus \left\{\left(c,\mathbf{0}\right)\mid c\in C\right\}\right).\end{equation} For example, Fujiwara \cite{Fujiwara2014} showed that the $\left[\!\left[7,1,3\right]\!\right]$ Steane code has a nonstandard choice of stabilizer generators that allows it to correct either a single data error or a single syndrome error, even though the standard choice of stabilizer generators cannot.

Not all stabilizer codes are able to correct measurement errors (see Section \ref{impureqdscodes}), but such codes can be augmented by redundant measurements. We call such a code with $l$ extra measurements beyond the $m$ required an $\left[\!\left[n,k,d:l\right]\!\right]$ quantum data-syndrome (QDS) code, and Fujiwara showed that by adding a single stabilizer measurement, any distance-3 stabilizer code can correct either a single data or syndrome error.

\begin{theorem}[{\cite[Theorem 1]{Fujiwara2014}}]
\label{fujiwarathm}
Any $\left[\!\left[n,k,3\right]\!\right]$ stabilizer code may be turned into an $\left[\!\left[n,k,3:1\right]\!\right]$ QDS code.
\end{theorem}

The extra stabilizer measurement in Theorem \ref{fujiwarathm} is the product of the $m$ stabilizer generators, and acts as a parity bit for the syndrome, effectively encoding the syndrome in a $\left[m+1,m,2\right]$ classical linear code. This guarantees that every data-syndrome error will have an even weight, leaving the weight-1 errors free for syndrome measurement errors. In \cite{Ashikhmin2020}, Ashikhmin, Lai, and Brun generalize this idea from parity check codes to binary linear codes.

Let $B$ be an $\left[m+l,m\right]$ binary linear code with generator matrix $G_{B}$ in systematic form \begin{equation}G_{B}=\begin{bmatrix}I_{m} & A\end{bmatrix}, \end{equation} and use it to pick $l$ codewords $\mathbf{b}_{j}$, $j\in\left[l\right]$ from our code $C$. In particular, we pick \begin{equation}\mathbf{b}_{j}=a_{1,j}\mathbf{g}_{1}+\cdots a_{m,j}\mathbf{g}_{m}.\end{equation} Crucially, each of these codewords is associated with an element of the stabilizer group, so such measurements do not collapse any superposition in the encoding subspace. We call the code $B$ the syndrome measurement (SM) code. If we let \begin{equation} A'^{T}=\begin{bmatrix} \mathbf{b}_{1}^{T} & \cdots & \mathbf{b}_{l}^{T} \end{bmatrix},\end{equation} we see that measuring the additional $l$ measurements results in a QDS code $\hat{C}$ generated by \begin{equation}\label{qdsgenmat} \hat{G}=\begin{bmatrix} G_{C} & I_{m} & \mathbf{0} \\ A' & \mathbf{0} & I_{l} \end{bmatrix}=\begin{bmatrix} G_{C} & I_{m} & \mathbf{0} \\ \mathbf{0} & A & I_{l} \end{bmatrix}.\end{equation}

While we can decode this code as a QDS code, it is often easier to first use a classical decoder to decode the SM code and then decode the stabilizer code. While this is usually suboptimal, it is typically much simpler, and so we take this approach for the remainder of the paper. After decoding the SM code and recovering the putative syndrome $\hat{\mathbf{s}}$, the probability of a syndrome error is given by \begin{equation}p_{se}=\Pr\!\left(\mathbf{s}\neq\hat{\mathbf{s}}\right),\end{equation} where $\mathbf{s}$ is the correct syndrome.

\section{QDS Subsystem Codes}

Here we extend the concept of QDS to subsystem codes. Starting with an $\left[\!\left[n,k,r,d'\right]\!\right]$ subsystem code defined by a pair of classical additive $\mathbb{F}_{4}$ codes $C, D$ such that $C$ is an $\left(n,2^{m-r}\right)_{4}$ code satisfying $C=D\cap D^{\perp}$ (where $D^{\perp}$ is the dual code of $D$ with respect to the $*$-inner product in Eq. \ref{traceinnerprod}), we pick a classical $\left[m-r+l,m-r\right]$ binary code to be our SM code.  Note that because only the stabilizer generators of the subsystem code will be measured and not the gauge operators, we have $r$ fewer measurements that need to be protected.

The generator matrix of the resultant QDS subsystem code $\hat{C}$ is identical to the one in (\ref{qdsgenmat}), and we identify the code by the parameters $\left[\!\left[n,k,r,d:l\right]\!\right]$, where the minimum distance of the code is \begin{equation}d=\min\wt\!\left(\hat{C}^{\perp}\setminus\left\{\left(c,\mathbf{0}\right)\mid c\in D\right\}\right)\leq d'.\end{equation} It is straightforward to see such a code can correct any combination of up to $t$ data and measurement errors, where $t=\left\lfloor\left(d-1\right)/2\right\rfloor$.

Much of the intuition for QDS stabilizer codes transfers to QDS subsystem codes. As an example, we generalize the result in Theorem \ref{fujiwarathm} to QDS subsystem codes:

\begin{theorem}\label{qdssubsystemd3}
Any $\left[\!\left[n,d,r,3\right]\!\right]$ subsystem code may be turned into an $\left[\!\left[n,k,r,3:1\right]\!\right]$ QDS subsystem code.
\end{theorem}
\begin{proof}
Let $\mathbf{g}_{1},\dots,\mathbf{g}_{m-r}$ be the codewords of $C$ associated with the stabilizer generators of the subsystem code. The additional stabilizer measurement will be the one associated with the codeword $\mathbf{g}_{1}+\cdots+\mathbf{g}_{m-r}$. It is easy to see that any syndrome due to a data error will have an even weight, and as a result all the weight-1 syndromes are available for the errors that occur due to a single faulty measurement.
\end{proof}

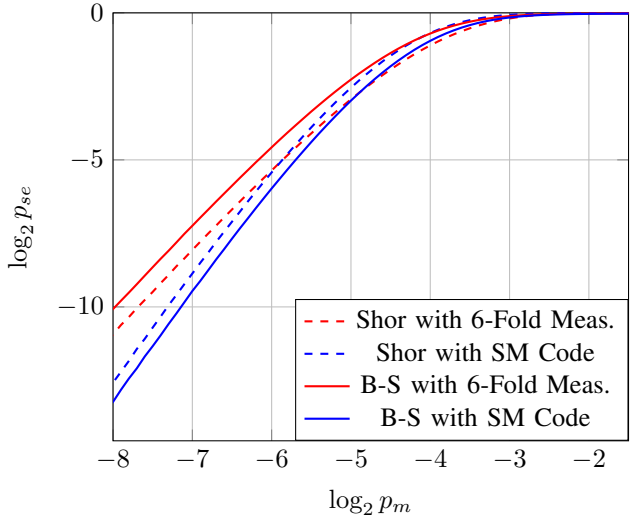
\begin{figure}[t]
\centering
\begin{tikzpicture}
\begin{axis}[
xlabel = {$\log_{2} p_{m}$},
xmajorgrids=true,
xmax = -1.5,
xmin = -8.0,
ylabel = {$\log_{2} p_{se}$},
ymajorgrids=true,
ymax = 0,
]
\addplot[mark = none, red, dashed, thick] table {
-1.5	-0.001055018
-1.6	-0.001595342
-1.7	-0.002391727
-1.8	-0.003696057
-1.9	-0.005567868
-2	-0.008362289
-2.1	-0.012268204
-2.2	-0.017799726
-2.3	-0.025265701
-2.4	-0.035520213
-2.5	-0.048861544
-2.6	-0.066114816
-2.7	-0.088090952
-2.8	-0.115506292
-2.9	-0.149251824
-3	-0.190335138
-3.1	-0.239068063
-3.2	-0.296357173
-3.3	-0.362343823
-3.4	-0.43859588
-3.5	-0.524307581
-3.6	-0.620113703
-3.7	-0.726196104
-3.8	-0.842463278
-3.9	-0.969368283
-4	-1.105477353
-4.1	-1.252153036
-4.2	-1.406830911
-4.3	-1.572588926
-4.4	-1.745131056
-4.5	-1.926664316
-4.6	-2.114137955
-4.7	-2.311934368
-4.8	-2.514726467
-4.9	-2.72436787
-5	-2.941235438
-5.1	-3.162140497
-5.2	-3.388097013
-5.3	-3.619963256
-5.4	-3.858513256
-5.5	-4.098249736
-5.6	-4.338978999
-5.7	-4.591686226
-5.8	-4.843181122
-5.9	-5.096458417
-6	-5.361886764
-6.1	-5.617617718
-6.2	-5.885296497
-6.3	-6.144891333
-6.4	-6.414503801
-6.5	-6.68724771
-6.6	-6.957033543
-6.7	-7.234634992
-6.8	-7.504622061
-6.9	-7.78348926
-7	-8.063906737
-7.1	-8.3505041
-7.2	-8.61698975
-7.3	-8.915337005
-7.4	-9.193700692
-7.5	-9.490496699
-7.6	-9.764297445
-7.7	-10.04317015
-7.8	-10.34527118
-7.9	-10.61546385
-8	-10.93787829
-8.1	-11.16864157
-8.2	-11.49362148
-8.3	-11.80047527
-8.4	-12.11428508
-8.5	-12.36795519
-8.6	-12.61741021
-8.7	-12.91702497
-8.8	-13.22108497
-8.9	-13.63093002
-9	-13.83897658
-9.1	-14.17397021
-9.2	-14.43097562
-9.3	-14.76258313
-9.4	-14.96939554
-9.5	-15.29449091
-9.6	-15.60964047
-9.7	-15.99861568
-9.8	-16.18337572
-9.9	-16.35697504
-10	-16.86817349
};
\addlegendentry{Shor with 6-Fold Meas.}
\addplot[mark = none, blue, dashed, thick] table {
-1.5	-0.001824754
-1.6	-0.002055614
-1.7	-0.002413347
-1.8	-0.002860836
-1.9	-0.003500468
-2	-0.004302258
-2.1	-0.005454556
-2.2	-0.006979154
-2.3	-0.009080187
-2.4	-0.011912167
-2.5	-0.015763506
-2.6	-0.0210183
-2.7	-0.028253793
-2.8	-0.037863011
-2.9	-0.050665728
-3	-0.067577999
-3.1	-0.089241319
-3.2	-0.117148142
-3.3	-0.152157202
-3.4	-0.195401841
-3.5	-0.248267031
-3.6	-0.311132035
-3.7	-0.385857408
-3.8	-0.472459032
-3.9	-0.572316063
-4	-0.684984522
-4.1	-0.812084619
-4.2	-0.95250698
-4.3	-1.106485009
-4.4	-1.274109865
-4.5	-1.456056437
-4.6	-1.648479262
-4.7	-1.85639985
-4.8	-2.076042205
-4.9	-2.30520392
-5	-2.546299709
-5.1	-2.798042353
-5.2	-3.057372898
-5.3	-3.328220904
-5.4	-3.606995743
-5.5	-3.893206142
-5.6	-4.18707673
-5.7	-4.489841911
-5.8	-4.797775069
-5.9	-5.111767202
-6	-5.427183074
-6.1	-5.75209071
-6.2	-6.08097271
-6.3	-6.414748989
-6.4	-6.761445457
-6.5	-7.103061701
-6.6	-7.445204888
-6.7	-7.790671413
-6.8	-8.154891373
-6.9	-8.501789988
-7	-8.865463608
-7.1	-9.223441161
-7.2	-9.589354974
-7.3	-9.945693286
-7.4	-10.33024753
-7.5	-10.70359546
-7.6	-11.08681377
-7.7	-11.43082003
-7.8	-11.83629518
-7.9	-12.2060958
-8	-12.57415402
-8.1	-12.97120456
-8.2	-13.36587144
-8.3	-13.77280111
-8.4	-14.23242994
-8.5	-14.54624539
-8.6	-14.89367848
-8.7	-15.23677541
-8.8	-15.70274988
-8.9	-15.98028385
-9	-16.57072149
-9.1	-16.69675114
-9.2	-17.20876104
-9.3	-17.66712597
-9.4	-17.80228555
-9.5	-18.28771238
-9.6	-18.90920076
-9.7	-19.28771238
-9.8	-19.56524636
-9.9	-20.70274988
-10	-19.96578428
};
\addlegendentry{Shor with SM Code}
\addplot[mark = none, red, thick] table {
-1.5	-0.020451791
-1.6	-0.020585209
-1.7	-0.020986621
-1.8	-0.021672042
-1.9	-0.022622011
-2	-0.024039715
-2.1	-0.026008586
-2.2	-0.02867919
-2.3	-0.032326215
-2.4	-0.037018823
-2.5	-0.043239311
-2.6	-0.051068884
-2.7	-0.061327294
-2.8	-0.0742278
-2.9	-0.090555031
-3	-0.110913771
-3.1	-0.135875156
-3.2	-0.166742082
-3.3	-0.203639821
-3.4	-0.247940746
-3.5	-0.300347043
-3.6	-0.361486396
-3.7	-0.431430102
-3.8	-0.511724812
-3.9	-0.601949909
-4	-0.702966823
-4.1	-0.814805217
-4.2	-0.936595211
-4.3	-1.069637339
-4.4	-1.212353493
-4.5	-1.365110556
-4.6	-1.528358435
-4.7	-1.701183532
-4.8	-1.883103605
-4.9	-2.070583925
-5	-2.267812219
-5.1	-2.472746961
-5.2	-2.684333787
-5.3	-2.902204244
-5.4	-3.126681398
-5.5	-3.354670129
-5.6	-3.591621682
-5.7	-3.830207252
-5.8	-4.071578508
-5.9	-4.319297576
-6	-4.57298658
-6.1	-4.831096468
-6.2	-5.085261369
-6.3	-5.349014918
-6.4	-5.610956398
-6.5	-5.879827246
-6.6	-6.14766151
-6.7	-6.416991269
-6.8	-6.691808069
-6.9	-6.963574034
-7	-7.240946458
-7.1	-7.513594533
-7.2	-7.805345566
-7.3	-8.091848079
-7.4	-8.36497933
-7.5	-8.650340325
-7.6	-8.937850651
-7.7	-9.228097523
-7.8	-9.513295923
-7.9	-9.790609023
-8	-10.07302899
-8.1	-10.35424301
-8.2	-10.64403654
-8.3	-10.98051157
-8.4	-11.24619011
-8.5	-11.52383609
-8.6	-11.8358831
-8.7	-12.10705726
-8.8	-12.43721297
-8.9	-12.73504328
-9	-13.00462403
-9.1	-13.33701069
-9.2	-13.61953651
-9.3	-13.85642373
-9.4	-14.1258247
-9.5	-14.51556279
-9.6	-14.73619636
-9.7	-15.00601413
-9.8	-15.51160839
-9.9	-15.74236173
-10	-15.93368344
};
\addlegendentry{B-S with 6-Fold Meas.}
\addplot[mark = none, blue, thick] table {
-1.5	-0.030128826
-1.6	-0.030403869
-1.7	-0.030996813
-1.8	-0.031976959
-1.9	-0.033285119
-2	-0.035320388
-2.1	-0.038186254
-2.2	-0.042172114
-2.3	-0.04729481
-2.4	-0.054057598
-2.5	-0.062808738
-2.6	-0.0739937
-2.7	-0.088418086
-2.8	-0.106548085
-2.9	-0.129365164
-3	-0.157580157
-3.1	-0.19249417
-3.2	-0.234480105
-3.3	-0.284947349
-3.4	-0.34510504
-3.5	-0.415640491
-3.6	-0.497542729
-3.7	-0.591576367
-3.8	-0.69755014
-3.9	-0.817279183
-4	-0.950823805
-4.1	-1.096749584
-4.2	-1.256574543
-4.3	-1.430241466
-4.4	-1.616401728
-4.5	-1.815629242
-4.6	-2.027271325
-4.7	-2.24977319
-4.8	-2.484384928
-4.9	-2.730830073
-5	-2.984832884
-5.1	-3.250424973
-5.2	-3.524741796
-5.3	-3.806778951
-5.4	-4.094937206
-5.5	-4.392865255
-5.6	-4.698960272
-5.7	-5.008811424
-5.8	-5.325759928
-5.9	-5.648523184
-6	-5.973015854
-6.1	-6.306373435
-6.2	-6.643869715
-6.3	-6.989025428
-6.4	-7.329101667
-6.5	-7.680012191
-6.6	-8.039559463
-6.7	-8.387283276
-6.8	-8.754966853
-6.9	-9.117912617
-7	-9.460369675
-7.1	-9.840009217
-7.2	-10.21920278
-7.3	-10.58290956
-7.4	-10.96849188
-7.5	-11.33549531
-7.6	-11.68646157
-7.7	-12.08902179
-7.8	-12.43571354
-7.9	-12.81945491
-8	-13.23026211
-8.1	-13.56088116
-8.2	-13.94288388
-8.3	-14.40409156
-8.4	-14.73804321
-8.5	-15.11458495
-8.6	-15.44976914
-8.7	-16.01867923
-8.8	-16.35036799
-8.9	-16.94857499
-9	-17.05505162
-9.1	-17.57401656
-9.2	-17.63694782
-9.3	-18.39240976
-9.4	-18.33165573
-9.5	-19.15020886
-9.6	-19.15020886
-9.7	-20.28771238
-9.8	-20.76164357
-9.9	-20.68364106
-10	-21.12421365
};
\addlegendentry{B-S with SM Code}
\end{axis}
\end{tikzpicture}
\caption{Probability of syndrome error $p_{se}$ as a function of single-qubit measurement error $p_{m}$ for Shor and Bacon-Shor codes protected by 6-fold measurements and SM codes.}
\label{smvsfoldplot}
\end{figure}

Using Shor's $l$-fold repeated syndrome extraction is equivalent to encoding each qubit individually in a length $l$ repetition code, which is much less efficient than encoding the syndrome bits together in a single classical code. In Fig. \ref{smvsfoldplot}, we show how the 9-qubit Shor stabilizer code \cite{Shor1995} and the 9-qubit Bacon-Shor subsystem code \cite{Bacon2006} compare using both a 6-fold repeated syndrome extraction and SM codes to protect against syndrome measurement errors. For the $X$-stabilizer generators for both the Shor and Bacon-Shor codes as well as the $Z$-generators for the Bacon-Shor code, we encode them in the optimal $\left[12,2,8\right]$ Cordaro-Wagner code \cite{Cordaro1967}. For the $Z$-stabilizer generators for the Shor code, we encode them in the optimal $\left[18,6,8\right]$ code found in Grassl's table of classical linear codes \cite{GrasslONLINE}. In each case, a total of 144 measurements were done.

Since subsystem codes typically have heavier stabilizer generators than their gauge-fixed stabilizer codes, each corresponding syndrome bit has a higher probability $p_{err}$ of being erroneous. Therefore, when an $l$-fold measurement scheme is used, the subsystem code will experience a higher effective error rate. On the other hand, when a SM code is used the subsystem code outperforms the stabilizer code. This is because even though in the stabilizer case multiple low-weight stabilizer elements are measured, the high-weight products of these elements are also measured. Since these elements are often just as heavy as the subsystem code's stabilizer generator, no advantage is gained. In addition, the stabilizer case requires a longer SM code to acheive the same minimum distance, which increases the number of syndrome bits on which a potential error might occur.

\begin{figure}
\centering
\begin{tikzpicture}
\begin{axis}[
xlabel = {$\log_{2} p_{m}$},
xmajorgrids=true,
xmax = -1.5,
xmin = -7,
ylabel = {$\log_{2} p_{se}$},
ymajorgrids=true,
ymax = 0,
]
\addplot[mark = none, red, dashed, thick] table {
-1.5	-0.005037931
-1.6	-0.005881374
-1.7	-0.006892345
-1.8	-0.008278061
-1.9	-0.010314578
-2	-0.012685956
-2.1	-0.016153891
-2.2	-0.020758029
-2.3	-0.027261114
-2.4	-0.035732496
-2.5	-0.047304466
-2.6	-0.062724766
-2.7	-0.083675933
-2.8	-0.11117427
-2.9	-0.146779037
-3	-0.192228431
-3.1	-0.250227091
-3.2	-0.321537887
-3.3	-0.408683741
-3.4	-0.513895861
-3.5	-0.637443321
-3.6	-0.780293232
-3.7	-0.9462902
-3.8	-1.132387039
-3.9	-1.339531965
-4	-1.570242618
-4.1	-1.819814337
-4.2	-2.094523939
-4.3	-2.387567878
-4.4	-2.695903354
-4.5	-3.029976962
-4.6	-3.37600149
-4.7	-3.743959301
-4.8	-4.119907149
-4.9	-4.519851229
-5	-4.918379156
-5.1	-5.346962919
-5.2	-5.785253477
-5.3	-6.216396674
-5.4	-6.677581472
-5.5	-7.13739785
-5.6	-7.610510281
-5.7	-8.095465806
-5.8	-8.589072818
-5.9	-9.081567966
-6	-9.583830429
-6.1	-10.10404599
-6.2	-10.61608286
-6.3	-11.13485716
-6.4	-11.65799149
-6.5	-12.21665231
-6.6	-12.75964903
-6.7	-13.28630419
-6.8	-13.86228664
-6.9	-14.37422454
-7	-14.99861568
};
\addlegendentry{204 Meas. Shor}
\addplot[mark = none, blue, dashed, thick] table {
-1.5	-0.003306769
-1.6	-0.003725946
-1.7	-0.004266197
-1.8	-0.005073842
-1.9	-0.006159518
-2	-0.007619312
-2.1	-0.009667663
-2.2	-0.012314754
-2.3	-0.016085219
-2.4	-0.021398212
-2.5	-0.029007477
-2.6	-0.039197276
-2.7	-0.053392235
-2.8	-0.072444532
-2.9	-0.098240687
-3	-0.132710757
-3.1	-0.176059936
-3.2	-0.232472979
-3.3	-0.302762385
-3.4	-0.388748669
-3.5	-0.494063034
-3.6	-0.616655562
-3.7	-0.762048682
-3.8	-0.925709725
-3.9	-1.113442537
-4	-1.32435482
-4.1	-1.555660916
-4.2	-1.81030193
-4.3	-2.084661957
-4.4	-2.380179125
-4.5	-2.693158144
-4.6	-3.025739153
-4.7	-3.376182875
-4.8	-3.745334157
-4.9	-4.124809715
-5	-4.520847061
-5.1	-4.933919423
-5.2	-5.353100051
-5.3	-5.790016444
-5.4	-6.243574443
-5.5	-6.700901905
-5.6	-7.166417658
-5.7	-7.637451561
-5.8	-8.122664081
-5.9	-8.613873324
-6	-9.120153107
-6.1	-9.620733934
-6.2	-10.13669584
-6.3	-10.64832645
-6.4	-11.18350678
-6.5	-11.71283507
-6.6	-12.24181624
-6.7	-12.76180012
-6.8	-13.34109847
-6.9	-13.84104259
-7	-14.39228857
};
\addlegendentry{216 Meas. Shor}
\addplot[mark = none, red, thick] table {
-1.5	-0.068708868
-1.6	-0.06939937
-1.7	-0.070589843
-1.8	-0.072885297
-1.9	-0.076123949
-2	-0.081223804
-2.1	-0.087683306
-2.2	-0.096479867
-2.3	-0.108314221
-2.4	-0.124161218
-2.5	-0.144181184
-2.6	-0.169625251
-2.7	-0.201848365
-2.8	-0.243150165
-2.9	-0.292978595
-3	-0.354749398
-3.1	-0.428492768
-3.2	-0.517614408
-3.3	-0.623512343
-3.4	-0.746703893
-3.5	-0.887335727
-3.6	-1.04683349
-3.7	-1.228397496
-3.8	-1.428661153
-3.9	-1.651620259
-4	-1.891943972
-4.1	-2.155823274
-4.2	-2.438786302
-4.3	-2.739585863
-4.4	-3.055053496
-4.5	-3.396071055
-4.6	-3.749923713
-4.7	-4.119189114
-4.8	-4.499483455
-4.9	-4.897704687
-5	-5.310396675
-5.1	-5.730739393
-5.2	-6.173709823
-5.3	-6.607242966
-5.4	-7.074547391
-5.5	-7.535473919
-5.6	-8.012775196
-5.7	-8.488560299
-5.8	-8.975578187
-5.9	-9.47495747
-6	-9.980739322
-6.1	-10.49485703
-6.2	-11.00416096
-6.3	-11.53777593
-6.4	-12.05026338
-6.5	-12.56388081
-6.6	-13.06815661
-6.7	-13.66536766
-6.8	-14.22248476
-6.9	-14.79971161
-7	-15.30403468

};
\addlegendentry{204 Meas. B-S}
\addplot[mark = none, blue, thick] table {
-1.5	-0.037689749
-1.6	-0.038380839
-1.7	-0.038825384
-1.8	-0.040196467
-1.9	-0.042196346
-2	-0.045363967
-2.1	-0.049479795
-2.2	-0.055644243
-2.3	-0.063426601
-2.4	-0.07386581
-2.5	-0.08754222
-2.6	-0.105654645
-2.7	-0.128733446
-2.8	-0.158515911
-2.9	-0.196758495
-3	-0.24414408
-3.1	-0.302526403
-3.2	-0.374318039
-3.3	-0.462267214
-3.4	-0.563646907
-3.5	-0.685039609
-3.6	-0.824682996
-3.7	-0.985401005
-3.8	-1.163714287
-3.9	-1.367389767
-4	-1.58915149
-4.1	-1.834073594
-4.2	-2.096959801
-4.3	-2.382774406
-4.4	-2.684520717
-4.5	-3.007166323
-4.6	-3.345454474
-4.7	-3.699559942
-4.8	-4.068453088
-4.9	-4.466396785
-5	-4.850155349
-5.1	-5.272663274
-5.2	-5.703524972
-5.3	-6.152195466
-5.4	-6.576260493
-5.5	-7.045383787
-5.6	-7.509889962
-5.7	-7.986410539
-5.8	-8.479230792
-5.9	-8.959531149
-6	-9.453083321
-6.1	-9.96321677
-6.2	-10.49134824
-6.3	-10.99529871
-6.4	-11.52217763
-6.5	-12.04835253
-6.6	-12.5925537
-6.7	-13.11338586
-6.8	-13.69734989
-6.9	-14.20194783
-7	-14.72333421
};
\addlegendentry{216 Meas. B-S}
\end{axis}
\end{tikzpicture}
\caption{Probability of a syndrome error for Shor and Bacon-Shor codes protected by the same SM code, but with the original stabilizer generators permuted, giving a different number of measurements.}
\label{diffmeasplot}
\end{figure}
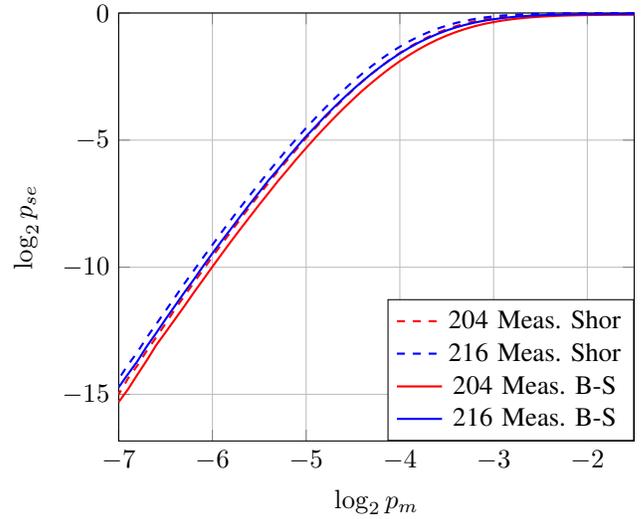

One interesting thing to keep in mind when choosing a SM code is that starting with different stabilizer generators, and in fact even permutations of the same set of stabilizer generators can produce a noticible difference in the probability of a syndrome error. This is due to stabilizer of different weights being measured in our SM scheme, resulting in a different number of total measurements being done. Ashikhmin et al. \cite{Ashikhmin2020} and Premakumar et al. \cite{Premakumar2021} sidestep this issue by choosing stabilizer code whose non-identity stabilizer elements are all of the same weight or of almost the same weights. While this is the case for the stabilizer group of the 9-qubit Bacon-Shor code, it is not true for the Shor code, which has $Z$-stabilizer generators that range from weight-2 to weight-6.

In Fig. \ref{diffmeasplot}, we use a $\left[25,6,11\right]$ SM code to protect the syndrome of the $Z$-stabilizers; however, by permuting our originally chosen stabilizer generators, we end up with two different QDS codes, one that does 102 total measurements for the $Z$-type syndrome, and another that makes 108 total measurements. We also protect the $X$-type syndrome and the $Z$-type syndromes of the Bacon-Shor code with $\left[17,2,11\right]$ and $\left[18,2,12\right]$ Cordaro-Wagner codes, so each protection scheme makes 204 or 216 total measurements.

In each case, the Bacon-Shor code maintains an advantage over the Shor code with the same number of measurements, but in both cases the SM codes that perform fewer total measurements have a noticeable advantage over the ones that perform more total measurements. In both cases, this happens because there is an increase in the total number of syndrome measurements occurring, without a corresponding increase in the minimum distance of the SM codes.

\section{Impure QDS Codes and New Bounds}\label{impureqdscodes}

In \cite{Fujiwara2014}, Fujiwara showed that the $\left[\!\left[7,1,3\right]\!\right]$ Steane code \cite{Steane1996} has a nonstandard choice of stabilizer generators that make it a $\left[\!\left[7,1,3:0\right]\!\right]$ QDS code (although this new set of stabilizer generators no longer maintain the CSS property), and asked more generally when an $\left[\!\left[n,k,d\right]\!\right]$ has a choice of stabilizer generators that make it an $\left[\!\left[n,k,d:0\right]\!\right]$ QDS code, so that it requires only $n-k$ total measurements. Intuitively, this should be more likely when the code has a larger number of available syndromes; conversely, it should be less likely when a code has fewer available syndromes for the syndrome measurement errors, such as the perfect stabilizer codes \cite{Li2013} which have no available syndromes, and are therefore more likely to need to make use of the construction from Theorem \ref{fujiwarathm} to protect against a faulty syndrome measurement.

\subsection{Impure QDS Codes}

A natural place to test this intuition is with impure quantum codes, as they make use of fewer syndromes than a pure code with the same parameters. Indeed we find this intuition to be correct and show in this section that any impure $\left[\!\left[n,k,3\right]\!\right]$ stabilizer code will always have a choice of stabilizer generators that make it an $\left[\!\left[n,k,3:0\right]\!\right]$ QDS code.

Recall that a code is pure if trace-orthogonal correctable errors map the coding subspace to orthogonal subspaces, and is impure otherwise. When restricted to stabilizer codes, this is the equivalent to the code being  nondegenerate and degenerate respectively. In particular, an impure code will have correctable errors that share syndromes, meaning there should be more freely-available syndromes for syndrome measurement errors than pure codes with the same parameters. A stabilizer code is impure if and only if there are non-identity stabilizer elements of weight less than the minimum distance.

There has historically been a strong connection between impure codes and subsystem codes, with the former inspiring the invention of the latter \cite{Bacon2006, Poulin2005}. Many subsystem codes have been constructed by gauging impure codes, and conversely many impure codes may be seen as gauge-fixings of subsystem codes.

\begin{example}\label{qds6ex}
We start with one of the smallest example of an impure distance-3 stabilizer code, the $\left[\!\left[6,1,3\right]\!\right]$ code investigated by Shaw et al. \cite{Shaw2008}, whose stabilizer group can be associated with the classical code generated by \begin{equation} G=\begin{bmatrix} 
0 & 0 & 0 & \omega & 0 & \omega \\
\omega^{2} & 0 & \omega & 1 & 1 & \omega^{2} \\
\omega & 1 & 0 & 0 & 1 & \omega \\
0 & \omega & 1 & 1 & 1 & 1 \\
\omega & \omega & \omega & 0 & \omega & 0 \\
\end{bmatrix}.\end{equation}

Similar to the construction in Theorems \ref{fujiwarathm} and \ref{qdssubsystemd3}, we will choose the stabilizer element associated with the classical $\mathbb{F}_{4}$ codeword $\mathbf{g}_{1}'=\mathbf{g}_{1}+\dots+\mathbf{g}_{5},$ which will act as a parity check for the syndrome, but instead of using it as a sixth measurement to make an overdetermined set of stabilizer elements, we will use it in place of the impure stabilizer generator associated with $\mathbf{g}_{1}$. Let $\omega_{j}^{i}\in\mathbb{F}_{4}^{6}$ denote the error that is $\omega^{i}$, $i\in\left\{0,1,2\right\}$ on a single coordinate $j\in\left[6\right]$ and $0$ at ever other coordinate. Every error of this form will have an even-weight syndrome, except for the four such that $\omega_{j}^{i}*\mathbf{g}_{1}=1$ (i.e., those whose associated Pauli error anticommutes with the weight-2 Pauli stabilizer generator associated with $\mathbf{g}_{1}$), and of these only the error $\omega_{4}^{0}$ has a weight-1 syndrome. Our goal is to therefore swap out some of the stabilizer generators with other stabilizer elements to increase the weight of $\omega_{4}^{0}$ while not adversely affecting the other syndromes.

Note that for any error $\omega_{j}^{i}$ such that $\omega_{j}^{i}*\mathbf{g}_{1}=0$, we have $\omega_{j}^{i}*\mathbf{g}_{\ell}=\omega_{j}^{i}*\left(\mathbf{g}_{\ell}+\mathbf{g}_{1}\right)$, $\ell\neq1$, meaning $\mathbf{g}_{\ell}$ can be replaced with $\mathbf{g}_{\ell}+\mathbf{g}_{1}$ and not affect the even-weight syndromes. Additionally, in order to preserve the linear independence of the rows of the generator matrix, we need to always change an even number of the rows in this way, as well as make sure that the resultant generator matrix $G'$ does not produce any weight-1 syndromes for the four errors that have odd-weight syndromes. Indeed, by replacing $\mathbf{g}_{4}$ and $\mathbf{g}_{5}$ with $\mathbf{g}_{4}+\mathbf{g}_{1}$ and $\mathbf{g}_{5}+\mathbf{g}_{1}$ respectively, we get a generator matrix that makes the stabilizer code a $\left[\!\left[6,1,3:0\right]\!\right]$ QDS code:
\begin{equation}G'=\begin{bmatrix} 
\omega^{2} & 1 & 1 & \omega & \omega^{2} & \omega \\
\omega^{2} & 0 & \omega & 1 & 1 & \omega^{2} \\
\omega & 1 & 0 & 0 & 1 & \omega \\
0 & \omega & 1 & \omega^{2} & 1 & \omega^{2} \\
\omega & \omega & \omega & \omega & \omega & \omega \\
\end{bmatrix}.\end{equation}

\end{example}

We now generalize the construction in Example \ref{qds6ex} and show that the construction can be applied to any impure distance-3 stabilizer code. But first, we need a brief detour to show that any distance-3 stabilizer code that is equivalent to a QDS code is itself a QDS code.

Two stabilizer codes $\mathcal{C}$ and $\mathcal{C}'$ are equivalent if there is a unitary map made up of local unitaries and qubit permutations such that maps encoded basis states of $\mathcal{C}$ to encoded basis states of $\mathcal{C}'$. Likewise, two additive codes $C$ and $C'$ over $\mathbb{F}_{4}$ are equivalent if there is a combination of codeword coordinate permutation, multiplication by nonzero scalars on fixed coordinates, and conjugation on fixed coordinates \cite{Calderbank1998}. Two stabilizer codes are equivalent if and only if the classical additive codes associated with their stabilizer groups are equivalent.

\begin{lemma}\label{equivqds}
Let $\mathcal{C}$ be an $\left[\!\left[n,k,3:0\right]\!\right]$ QDS code. Then any stabilizer code $\mathcal{C}'$ equivalent to $\mathcal{C}$ has a choice of stabilizer generators that makes it an $\left[\!\left[n,k,3:0\right]\!\right]$ QDS code.
\end{lemma}
\begin{proof}
Let $C$ and $C'$ be the equivalent classical codes associated with the stabilizer groups of $\mathcal{C}$ and $\mathcal{C}'$. Let $G$ be the generator matrix for $C$, with the rows $g_{i}$ associated with the choice of stabilizer generators that make $\mathcal{C}$ a QDS code, and let $s_{i,j}$, $i\in\left\{0,1,2\right\}$ be the syndrome for the error $\omega^{i}$ on the coordinate $j\in\left[n\right]$. A permutation $\sigma$ on the coordinates will take $s_{i,j}$ to $s_{i,\sigma\left(j\right)}$, so the set of syndromes is stabilized by permutations. Multiplication by $\omega$ (respectively, $\omega^{2}$) on a coordinate will cycle the syndromes on that coordinate, so $s_{i,j}$ becomes $s_{i+1\!\!\!\mod 3,j}$ (respectively, $s_{i+2\!\!\!\mod 3,j}$), both of which stabilize the set of syndromes. Finally, conjugation on a coordinate $j$ swaps the syndromes $s_{1,j}$ and $s_{2,j}$, again stabilizing the set of syndromes. If we take the generator matrix $G$ and apply the necessary transformations to make it the generator matrix $G'$ for the code $C'$, the set of syndromes has been preserved and therefore the rows of $G'$ are a valid choice of stabilizer generators that make $\mathcal{C}'$ an $\left[\!\left[n,k,3:0\right]\!\right]$ QDS code.
\end{proof}

\begin{theorem}\label{d3qdsconstruction}
Any impure $\left[\!\left[n,k,3\right]\!\right]$ stabilizer code $\mathcal{C}$ has a choice of stabilizer generators that make it an $\left[\!\left[n,k,3:0\right]\!\right]$ QDS code.
\end{theorem}
\begin{proof}
Let $C$ be the classical code associated with the stabilizer group $\mathcal{S}$ of $\mathcal{C}$ and let $G$ be a generator matrix with rows $g_{i}$, such that $\wt\!\left(g_{1}\right)<3$. Let $S_{i}$ be the stabilizer generator associated with $g_{i}$. We will slightly modify the approach from Theorems \ref{fujiwarathm} and \ref{qdssubsystemd3}, and replace row $g_{1}$ with $g_{1}'=g_{1}\oplus g_{2}\oplus\cdots\oplus g_{m}$. Any single-qubit Pauli error $\omega^{i}$ such that $\omega^{i}*g_{1}=0$ (i.e., their associated Pauli operators commute) will have an even-weight syndrome, as the syndrome bit associated with $g_{1}'$ will act as a parity bit for the rest of the syndrome. If $\omega^{i}*g_{1}=1$ (i.e., their associated Pauli operators anticommute), then the syndrome will be odd-weight. 

For $\mathcal{C}$ to be an $\left[\!\left[n,k,3:0\right]\!\right]$ QDS code, we need to make sure there are no syndromes of weight 1. All of the potentially problematic syndromes are from $\omega^{i}$ such that $\omega^{i}*g_{1}=1$. We can change the syndromes of these errors while keeping the syndromes for all other errors the same by adding $g_{1}$ to an even number of the $m$ generator matrix rows (as adding $g_{1}$ to an odd number of rows would reduce the rank of the generator matrix by one). This is equivalent to adding an even-weight bitstring $a$ to the odd-weight syndromes.

Consider the case where $\wt\!\left(g_{1}\right)=1$, so there are two single-qubit Pauli errors such that $\omega^{i}*g_{1}=1$, so there are two odd-weight syndromes $s_{1}$, $s_{2}$. In the worst possible case we need there to be more than $2m$ even-weight bitstrings to guarantee there is one $a$ such that $s_{1}\oplus a$ and $s_{2}\oplus a$ both do not have weight 1. Since there are $2^{m-1}$ possible even-weight strings of length $m$, we are guaranteed to have such a string whenever $m<2^{m-2}$, which is always true when $m>4$. Since any distance-3 stabilizer code other than the pure $\left[\!\left[5,1,3\right]\!\right]$ code \cite{Bennett1996, Laflamme1996} (which is unique up to equivalence of stabilizer codes) necessarily has at least 5 stabilizer generators, it follows that we can always choose a subset of the $\left\{g_{1}', g_{2}, \dots, g_{m}\right\}$ that we can add $g_{1}$ to that gives an $\left[\!\left[n,k,3:0\right]\!\right]$ QDS code.

Consider the case where $\wt\!\left(g_{1}\right)=2$. Here, there are four single-qubit Pauli errors such that $\omega^{i}*g_{1}=1$, so there are four odd-weight syndromes $s_{i}$, $i\in\left[4\right]$. In the worst possible case we need there to be more than $4m$ even-weight bitstrings to guarantee there is one $a$ such that $s_{i}\oplus a$ does not have weight 1 for all $i\in\left[4\right]$. Since there are $2^{m-1}$ possible even-weight strings of length $m$, we are guaranteed to have such a string whenever $m<2^{m-3}$, which is always true when $m>5.6$. All impure distance 3 stabilizer codes must obey the quantum Hamming bound \cite{Gottesman1997}: \begin{equation}m\geq\left\lceil\log_{2}\!\left(3n+1\right)\right\rceil,\end{equation} so we only need to check for stabilizer codes of length under 11 with fewer than 6 stabilizer generators. From Grassl's code table \cite{GrasslONLINE} the only code parameters that satisfy these restrictions are $\left[\!\left[5,1,3\right]\!\right]$, $\left[\!\left[6,1,3\right]\!\right]$, and $\left[\!\left[8,3,3\right]\!\right]$. As mentioned above, all $\left[\!\left[5,1,3\right]\!\right]$ stabilizer codes must be pure, and similarly all $\left[\!\left[8,3,3\right]\!\right]$ stabilizer codes are equivalent \cite{Cross2022} to the pure code discovered by Gottesman \cite{Gottesman1996}. There are two equivalency classes of $\left[\!\left[6,1,3\right]\!\right]$ stabilizer codes \cite{Calderbank1998}, both of which are impure. The first of these has a stabilizer element of weight 1, which is covered by the previous case, while Example \ref{qds6ex} and Lemma \ref{equivqds} together show that any stabilizer code from the second equivalency class has a choice of stabilizer generators that make it an $\left[\!\left[n,k,3:0\right]\!\right]$ QDS code. Therefore every distance 3 code where $\wt\!\left(g_{1}\right)=2$ has a choice of stabilizer generators that make it an $\left[\!\left[n,k,3:0\right]\!\right]$ QDS code, finishing the proof.
\end{proof}

Even though the search space of even-weight strings is exponential in $m$, this search can be done efficiently, as only $O\!\left(m\right)$ strings need to be checked in the worst case before a valid string is found.

\subsection{QDS Hamming Bound}

Both Fujiwara \cite{Fujiwara2014} and Ashikhmin, Lai, and Brun \cite{Ashikhmin2020} showed that any pure $\left[\!\left[n,k,d:r\right]\!\right]$ QDS code with minimum distance $d=2t+1$ must satisfy the Hamming bound for QDS codes: \begin{equation} \sum\limits_{i=0}^{t}\binom{n}{i}3^{i}\sum\limits_{j=0}^{t-i}\binom{m}{j}\leq 2^{m}. \end{equation} This bound also holds for impure QDS codes if the length $n$ is large enough.

\begin{theorem}\label{qdshammingthm}
All $\left[\!\left[n,k,3:0\right]\!\right]$ QDS codes satisfy the Hamming bound for QDS codes.
\end{theorem}
\begin{proof}
As this was already shown to hold for pure codes in \cite{Ashikhmin2020}, we only need to show it holds for impure codes as well. Our proof follows the proof by Gottesman \cite{Gottesman1997} that the quantum Hamming bound holds for any impure $\left[\!\left[n,k,3\right]\!\right]$ stabilizer code. Note that in the special case when $d=3$, the Hamming bound for QDS codes reduces to \begin{equation}\label{simpleqdshamming}4n-k+1\leq 2^{n-k}.\end{equation}

Any coordinate with a weight-1 stabilizer element acting on it may be deleted, ending up with an $\left[\!\left[n-1,k,3\right]\!\right]$ code (see \cite[Theorem 6]{Calderbank1998}). If this shortened code satisfies the inequality in Eq. \ref{simpleqdshamming}, then we have \begin{align*}
2^{n-k} & \geq 2\left(4\left(n-1\right)-k+1\right) \\
 & \geq 8n-2k-6 \\
 & \geq 4n-k+1,
\end{align*}
where the last inequality is due to the Singleton bound for QDS codes \cite[Theorem 5]{Ashikhmin2020}, so the original will also satisfy the inequality in Eq. \ref{simpleqdshamming}. Therefore, we need to only look at codes with weight-2 stabilizer elements.

Suppose the code has $\ell$ independent weight-2 stabilizer elements $M_{1},\dots,M_{\ell}$, and let $\mathcal{D}=\left\langle M_{1},\dots,M_{\ell}\right\rangle$ be the subgroup of $\mathcal{S}$ generated by them. Any Pauli operator in the centralizer $C_{G_{n}}\!\!\left(\mathcal{D}\right)$ of $\mathcal{D}$ will take states stabilized by $\mathcal{D}$ to states stabilized by $\mathcal{D}$, and the subspace stabilized by $\mathcal{D}$ has dimension $2^{n-\ell}$. Any single-qubit Pauli operator that acts on a qubit not acted upon by $\mathcal{D}$ is in $C_{G_{n}}\!\!\left(\mathcal{D}\right)$, and will map the subspace stabilized by $\mathcal{D}$ to an orthogonal subspace. As there are always at least $n-2\ell$ qubits not acted on by $\mathcal{D}$, we get the bound \begin{equation}\label{impurebound} 3\left(n-2\ell\right)+1\leq 2^{n-k-\ell}.\end{equation}

The inequality in Eq. \ref{impurebound} is stricter than the one in Eq. \ref{simpleqdshamming} whenever \begin{equation} \ell\geq\log_{2}\left(\frac{4n-k+1}{3\left(n-2\ell\right)+1}\right),\end{equation} and we can simplify and say that the inequality in Eq. \ref{impurebound} is stricter than the one in Eq. \ref{simpleqdshamming} whenever \begin{align} \ell & \geq\log_{2}\left(\frac{4n+4/3}{3\left(n-2\ell\right)+1}\right) \\ & = \log_{2}\left(\frac{4}{3}+\frac{8\ell}{3\left(n-2\ell\right)+1}\right).\label{bottomex}\end{align} Assume that $n\geq 2\ell$. Then Eq. \ref{simpleqdshamming} holds when $\ell\geq\log_{2}\left(4/3+8\ell\right)$, which happens whenever $\ell\geq 6$. For $\ell=5$, Eq. \ref{bottomex} holds for $n\geq 11$, when $\ell=4$ it holds for $n\geq 9$, when $\ell=3$ it holds for $n\geq 7$, when $\ell=2$ it holds for $n\geq 6$, and when $\ell=1$ it holds for $n\geq 6$. Then remaining cases when $n\geq 2\ell$ 
can all be ruled out using Eq. \ref{simpleqdshamming} and Grassl's code table \cite{GrasslONLINE}. Now assume $n<2\ell$. Then $k\leq n-\ell\leq n/2$, which is stricter than \begin{equation}k\leq n-\log_{2}\!\left(4n+1\right) \end{equation} (which is in turn stricter than the inequality in Eq. \ref{simpleqdshamming}) whenever $n\geq 11$, and similar to above the remaining cases can all be ruled out using Eq. \ref{simpleqdshamming} and Grassl's code table.
\end{proof}

\subsection{A New Bound for Impure Codes}

As a direct result of combining Theorems \ref{d3qdsconstruction} and \ref{qdshammingthm}, we obtain a new bound on single-error correcting impure stabilizer codes:

\begin{corollary}\label{newimpbound}
Any impure $\left[\!\left[n,k,3\right]\!\right]$ stabilizer code satisfies the bound \begin{equation} \log_{2}\!\left(4n-k+1\right)\leq n-k.\label{impbound}\end{equation}
\end{corollary}

Used in combination with the bound given by Yu et al. in \cite[Theorem 1]{Yu2013}, Eq. \ref{impbound} allows us to determine parameters where pure stabilizer codes are known to exist but impure stabilizer codes do not. For the next two corollaries, let $f_{a}=\left(4^{a}-1\right)/3$.

\begin{corollary}\label{cor1}
There exists pure, but not impure, stabilizer codes with parameters \begin{equation}\left[\!\left[f_{a}-\ell,f_{a}-\ell-\left\lceil\log_{2}\!\left(3\left(f_{a}-\ell\right)+1\right)\right\rceil,3\right]\!\right],\end{equation} where $\ell=0, 4, 5,\dots$, $\ell<2a/3$, $a\geq 2$.
\end{corollary}

\begin{corollary}\label{cor2}
There exists pure, but not impure, stabilizer codes with parameters \begin{equation}\left[\!\left[8f_{a}-\ell,8f_{a}-\ell-\left\lceil\log_{2}\!\left(3\left(8f_{a}-\ell\right)+1\right)\right\rceil,3\right]\!\right],\end{equation} where $\ell=0, 2, 3,\dots$, $\ell<\left(2a-4\right)/3$, $a\geq 1$.
\end{corollary}

In particular, there are no $\left[\!\left[n,k,3:0\right]\!\right]$ QDS codes for the parameters given in Corollaries \ref{cor1} and \ref{cor2}, as such code are necessarily pure and violate the Hamming bound for QDS codes.

\begin{figure}[t]
\begin{tikzpicture}
\begin{axis}[
xlabel = {$n$},
xmajorgrids=true,
xmax = 26.1,
xmin = 18.9,
ylabel = {$k$},
ymajorgrids=true,
ymax = 19.1,
ymin = 11.9,
domain = 18.9:26.1
]
\addplot[solid, black, thick] {x-4};
\addplot[solid, blue, thick] {x - (ln(1 + 3 * x) / ln(2))};
\addplot[solid, red, thick] table{
18.9	12.9069
20 	13.9325
21	14.868
22	15.8064
23	16.7473
24	17.6906
25	18.6361
26	19.5836
};
\addplot[dashed, red, thick] {x - 1 - (ln(3 * x - 5) / ln(2))};
\addplot[mark=*] coordinates {(19,12)};
\addplot[mark=*] coordinates {(20,13)};
\addplot[mark=*] coordinates {(21,15)};
\addplot[mark=*] coordinates {(22,15)};
\addplot[mark=*] coordinates {(23,16)};
\addplot[mark=*] coordinates {(24,17)};
\addplot[mark=*] coordinates {(25,18)};
\addplot[mark=*] coordinates {(26,19)};
\end{axis}
\end{tikzpicture}
\caption{Upper bounds on quantum stabilizer codes of minimum distance 3. The quantum Singleton (black) and Hamming (blue) bounds apply to both pure and impure codes, while the bound in Corollary \ref{newimpbound} (red) and the conjectured bound in Conjecture \ref{boundconj} (red, dashed) apply only to impure codes. Marked points denote optimal parameters \cite{GrasslONLINE} for codes of length $n$.}
\label{qhamplot}
\end{figure}
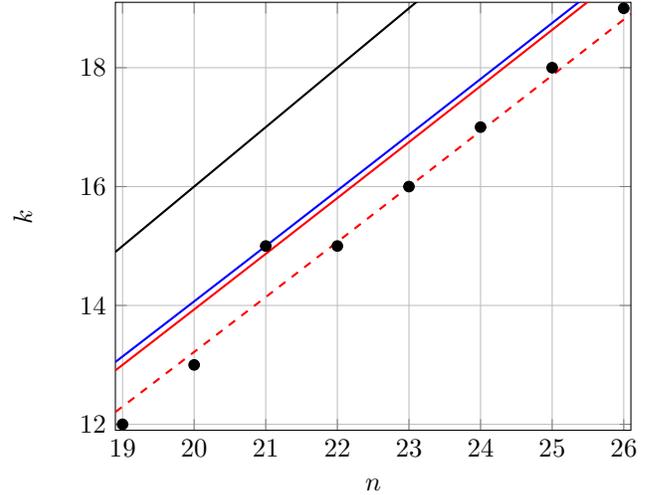

To our knowledge, most work on bounding impure quantum codes has focused on showing that every impure quantum code must satisfy the quantum Hamming bound (see \cite{Gottesman1996, Li2010, Sarvepalli2010}), most recently by Dallas et al. \cite{Dallas2022} who extended the quantum Hamming bound to all impure binary codes with $d<127$. In contrast, this new bound in Corollary \ref{newimpbound} is to our knowledge the first one that is satisfied by impure stabilizer codes but is violated by pure stabilizer codes, with the exception of the quantum Singleton bound which can be achieved by pure stabilizer codes (quantum maximum distance separable codes) but not impure stabilizer codes \cite{Calderbank1998, Rains1999}. This seems to go against the intuition that an impure code should be able to encode more information than a pure code due to syndromes used by multiple errors.

A perhaps better way to understand this is that high-weight stabilizer generators can detect more errors than low-weight ones, as all errors will commute with the large number of identity tensor components in the low-weight stabilizer. This likely overshadows the small advantage impure codes have of having multiple errors share syndromes, which can require a significant number of independent low-weight stabilizer elements.

Interestingly, there seems to be a large gap between the new bound in Corollary \ref{newimpbound} and the best known impure codes, such as the impure $\left[\!\left[22,15,3\right]\!\right]$ code marked in Fig. \ref{qhamplot}. Inspired by the inequality in Eq. \ref{impurebound} used in the proof of Theorem \ref{qdshammingthm}, as well as our intuition that the best distance-3 impure code should have exactly one weight-2 stabilizer element, we conjecture a stronger bound holds:

\begin{conjecture}\label{boundconj}
Any impure $\left[\!\left[n,k,3\right]\!\right]$ stabilizer code satisfies the bound \begin{equation} \log_{2}\!\left(3\left(n-2\right)+1\right)\leq \left(n-1\right)-k.\label{conjecturebound}\end{equation}
\end{conjecture}

\section{Conclusion}

In this paper we extend the notion of QDS codes from stabilizer codes to subsystem codes. Since subsystem codes have fewer stabilizer generators than equivalent stabilizer codes derived from gauge fixing, they can be encoded in shorter classical SM codes. Additionally, while the gauge-fixed stabilizer codes often have a larger number of low-weight stabilizer elements that result in syndrome bits with a lower $p_{err}$ than those of the subsystem code, they also have a larger number of high-weight stabilizer elements that negate this advantage.

We note that the SM codes we use were chosen because they were optimal classical codes, but they may not necessarily be optimal when viewed as SM codes, as the effective channel that needs to be protected against is not i.i.d unless all the non-identity stabilizer elements have the same weight, such as the Steane code investigated in \cite{Ashikhmin2020}. As such, one of the biggest open questions in designing QDS codes is how to choose a good SM code for a specific stabilizer code with stabilizer elements of varied weights.

In addition to QDS subsystem codes, we showed that every impure distance-3 stabilizer code has a choice of stabilizer generators that make it a QDS code requiring no extra measurements, as well as show that they obey the QDS variant of the Hamming bound. Intuitively, we might expect this to be because impure codes can more efficiently pack syndromes than pure codes, which seems to be true with regards to reserving syndromes for syndrome measurement errors. On the other hand, we show a new bound for impure distance-3 codes follows directly from this variant of the Hamming bound, and we conjecture an even stronfer bound, suggesting that when it comes to packing syndromes for data errors, that pure codes may be more efficient than impure codes, which is supported by recent results that all stabilizer codes with $d<127$ must obey the quantum Hamming bound. It would be interesting to see if tighter bounds on impure codes can be attained for $d>3$ or for nonbinary stabilizer codes.

\section*{Acknowledgements}
The author would like to thank Narayanan Rengaswamy for introducing him to the problem of syndrome measurement errors, as well as Ken Brown and Andrew Cross for fruitful discussions. This work was supported by the Office of the Director of National Intelligence, Intelligence Advanced Research Projects Activity through ARO Contract W911NF-16-1-0082.

\bibliographystyle{IEEEtran}

\begin{thebibliography}{10}
\providecommand{\url}[1]{#1}
\csname url@samestyle\endcsname
\providecommand{\newblock}{\relax}
\providecommand{\bibinfo}[2]{#2}
\providecommand{\BIBentrySTDinterwordspacing}{\spaceskip=0pt\relax}
\providecommand{\BIBentryALTinterwordstretchfactor}{4}
\providecommand{\BIBentryALTinterwordspacing}{\spaceskip=\fontdimen2\font plus
\BIBentryALTinterwordstretchfactor\fontdimen3\font minus
  \fontdimen4\font\relax}
\providecommand{\BIBforeignlanguage}[2]{{%
\expandafter\ifx\csname l@#1\endcsname\relax
\typeout{** WARNING: IEEEtran.bst: No hyphenation pattern has been}%
\typeout{** loaded for the language `#1'. Using the pattern for}%
\typeout{** the default language instead.}%
\else
\language=\csname l@#1\endcsname
\fi
#2}}
\providecommand{\BIBdecl}{\relax}
\BIBdecl

\bibitem{Ashikhmin2014}
A.~Ashikhmin, C.-Y. Lai, and T.~A. Brun, ``{Robust Quantum Error Syndrome
  Extraction by Classical Coding},'' in \emph{Proceedings of the 2014 IEEE
  International Symposium on Information Theory (ISIT)}, Honolulu, Hawaii, USA,
  Jun. 2014, pp. 546--550.

\bibitem{Ashikhmin2016}
------, ``{Correction of Data and Syndrome Errors by Stabilizer Codes},'' in
  \emph{Proceedings of the 2016 IEEE International Symposium on Information
  Theory (ISIT)}, Barcelona, Spain, Jul. 2016, pp. 2274--2278.

\bibitem{Ashikhmin2020}
------, ``{Quantum Data Syndrome Codes},'' \emph{IEEE Journal on Selected Areas
  in Communications}, vol.~38, no.~3, pp. 449--462, 2020.

\bibitem{Bacon2006}
D.~Bacon, ``{Operator quantum error-correcting subsystems for self-correcting
  quantum memories},'' \emph{Phys. Rev. A}, vol.~73, no.~1, p. 012340, 2006.

\bibitem{Bennett1996}
C.~H. Bennett, D.~P. DiVincenzo, J.~A. Smolin, and W.~K. Wootters,
  ``{Mixed-state entanglement and quantum error correction},'' \emph{Phys. Rev.
  A}, vol.~54, no.~5, pp. 3824--3851, 1996.

\bibitem{Bombin2015}
H.~Bomb{\'i}n, ``{Single-Shot Fault-Tolerant Quantum Error Correction},''
  \emph{Phys. Rev. X}, vol.~5, no.~3, p. 031043, 2015.

\bibitem{Brown2016}
B.~J. Brown, N.~H. Nickerson, and D.~E. Browne, ``{Fault-tolerant error
  correction with the gauge color code},'' \emph{Nature Communications},
  vol.~7, no.~1, p. 12302, 2016.

\bibitem{Calderbank1998}
A.~R. Calderbank, E.~M. Rains, P.~W. Shor, and N.~J.~A. Sloane, ``{Quantum
  Error Correction Via Codes Over GF(4)},'' \emph{IEEE Transactions on
  Information Theory}, vol.~44, no.~4, pp. 1369--1387, 1998.

\bibitem{Campbell2019}
E.~T. Campbell, ``{A theory of single-shot error correction for adversarial
  noise},'' \emph{Quantum Science and Technology}, vol.~4, no.~2, p. 025006,
  2019.

\bibitem{Chao2018}
R.~Chao and B.~W. Reichardt, ``{Quantum Error Correction with Only Two Extra
  Qubits},'' \emph{Physical Review Letters}, vol. 121, no.~5, p. 050502, 2018.

\bibitem{Cordaro1967}
J.~T. Cordaro and T.~J. Wagner, ``{Optimum $\left(n,2\right)$ Codes for Small
  Values of Channel Error Probability},'' \emph{IEEE Transactions on
  Information Theory}, vol.~13, no.~2, pp. 349--350, 1967.

\bibitem{Cross2022}
A.~Cross, private communication, Sep. 2022.

\bibitem{Dallas2022}
E.~Dallas, F.~Andreadakis, and D.~Lidar, ``{No
  $\left(\left(n,K,d<127\right)\right)$ code can violate the quantum Hamming
  bound},'' 2022, arXiv:2208.11800v2 [quant-ph].

\bibitem{Delfosse2020}
N.~Delfosse and B.~W. Reichardt, ``{Short Shor-style syndrome sequences},''
  2020, arXiv:2008.05051 [quant-ph].

\bibitem{DiVincenzo1996}
D.~P. DiVincenzo and P.~W. Shor, ``{Fault-Tolerant Error Correction with
  Efficient Quantum Codes},'' \emph{Physical Review Letters}, vol.~77, no.~15,
  pp. 3260--3263, 1996.

\bibitem{Fujiwara2014}
Y.~Fujiwara, ``Ability of stabilizer quantum error correction to protect itself
  from its own imperfection,'' \emph{Physical Review A}, vol.~90, no.~6, p.
  062304, 2014.

\bibitem{Gottesman1996}
D.~Gottesman, ``{Class of quantum error-correcting codes saturating the quantum
  Hamming bound},'' \emph{Physical Review A}, vol.~54, no.~3, pp. 1862--1868,
  1996.

\bibitem{Gottesman1997}
\BIBentryALTinterwordspacing
------, ``{Stabilizer Codes and Quantum Error Correction},'' Ph.D.
  dissertation, California Institute of Technology, Pasadena, CA, 1997.
  [Online]. Available: \url{arXiv:quant-ph/9705052}
\BIBentrySTDinterwordspacing

\bibitem{GrasslONLINE}
\BIBentryALTinterwordspacing
M.~Grassl. {Bounds on the Minimum Distance of Linear Codes and Quantum Codes}.
  Accessed: Jun. 13, 2020. [Online]. Available: \url{http://www.codetables.de/}
\BIBentrySTDinterwordspacing

\bibitem{Kribs2013}
D.~Kribs and D.~Poulin, ``{Operator quantum error correction},'' in
  \emph{{Quantum Error Correction}}, D.~A. Lidar and T.~A. Brun, Eds.\hskip 1em
  plus 0.5em minus 0.4em\relax New York: Cambridge University Press, 2013.

\bibitem{Kuo2021}
K.-Y. Kuo, I.-C. Chern, and C.-Y. Lai, ``{Decoding of Quantum Data-Syndrome
  Codes via Belief Propagation},'' in \emph{Proceedings of the 2021 IEEE
  International Symposium on Information Theory (ISIT)}, Melbourne, Australia,
  Jul. 2021, pp. 1552--1557.

\bibitem{Laflamme1996}
R.~Laflamme, C.~Miquel, J.~P. Paz, and W.~H. Zurek, ``{Perfect Quantum Error
  Correcting Code},'' \emph{Phys. Rev. Lett.}, vol.~77, no.~1, pp. 198--201,
  1996.

\bibitem{Li2010}
Z.~Li and L.~Xing, ``{On a Problem Concerning the Quantum Hamming Bound for
  Impure Quantum Codes},'' \emph{IEEE Transactions on Information Theory},
  vol.~56, no.~9, pp. 4731--4734, 2010.

\bibitem{Li2013}
------, ``{Classification of $q$-Ary Perfect Quantum Codes},'' \emph{IEEE
  Transactions on Information Theory}, vol.~59, no.~1, pp. 631--634, 2013.

\bibitem{Poulin2005}
D.~Poulin, ``{Stabilizer Formalism for Operator Quantum Error Correction},''
  \emph{Phys. Rev. Lett.}, vol.~95, no.~23, p. 230504, 2005.

\bibitem{Premakumar2021}
V.~N. Premakumar, H.~Sha, D.~Crow, E.~Bach, and R.~Joynt, ``{2-designs and
  redundant syndrome extraction for quantum error correction},'' \emph{Quantum
  Information Processing}, vol.~20, no.~3, p.~84, 2021.

\bibitem{Rains1999}
E.~M. Rains, ``{Nonbinary Quantum Codes},'' \emph{IEEE Trans. Inform. Theory},
  vol.~45, no.~6, pp. 1827--1832, 1999.

\bibitem{Raveendran2022}
N.~Raveendran, N.~Rengaswamy, A.~K. Pradhan, and B.~Vasi{\'c}, ``{Soft Syndrome
  Decoding of Quantum LDPC Codes for Joint Correction of Data and Syndrome
  Errors},'' in \emph{Proceedings of the 2022 IEEE International Conference on
  Quantum Computing and Engineering (QCE)}, Broomfield, Colorado, USA, Sep.
  2022, pp. 275--281.

\bibitem{Sarvepalli2010}
P.~Sarvepalli and A.~Klappenecker, ``{Degenerate Quantum Codes and the Quantum
  Hamming Bound},'' \emph{Phys. Rev. A}, vol.~81, no.~3, p. 032318, 2010.

\bibitem{Shaw2008}
B.~Shaw, M.~M. Wilde, O.~Oreshkov, I.~Kremsky, and D.~A. Lidar, ``{Encoding one
  logical qubit into six physical qubits},'' \emph{Phys. Rev. A}, vol.~78,
  no.~1, p. 012337, 2008.

\bibitem{Shor1995}
P.~W. Shor, ``{Scheme for reducing decoherence in quantum computer memory},''
  \emph{Phys. Rev. A}, vol.~52, no.~4, pp. R2493--R2496, 1995.

\bibitem{Shor1996}
------, ``{Fault-tolerant quantum computation},'' in \emph{Proceedings of the
  37th IEEE Symposium on Foundations of Computer Science (FOCS)}, Burlington,
  Vermont, USA, Oct. 1996, pp. 56--67.

\bibitem{Steane1996}
A.~Steane, ``{Multiple-particle interference and quantum error correction},''
  \emph{Proceedings of the Royal Society of London A}, vol. 452, no. 1954, pp.
  2551--2577, 1996.

\bibitem{Steane1997}
------, ``{Active Stabilization, Quantum Computation, and Quantum State
  Synthesis},'' \emph{Phys. Rev. Lett.}, vol.~78, no.~11, pp. 2252--2255, 1997.

\bibitem{Tansuwannont2022}
T.~Tansuwannont and K.~R. Brown, ``{Adaptive syndrome measurements for
  Shor-style error correction},'' 2022, arXiv:2208.05601 [quant-ph].

\bibitem{Wagner2022}
T.~Wagner, H.~Kampermann, D.~Bru{\ss}, and M.~Kliesch, ``{Learning logical
  quantum noise in quantum error correction},'' 2022, arXiv:2209.09267
  [quant-ph].

\bibitem{Yu2013}
S.~Yu, J.~Bierbrauer, Y.~Dong, Q.~Chen, and C.~H. Oh, ``{All the Stabilizer
  Codes of Distance 3},'' \emph{IEEE Transactions on Information Theory},
  vol.~59, no.~8, pp. 5179--5185, 2013.

\bibitem{Zeng2019}
W.~Zeng, A.~Ashikhmin, M.~Woolls, and L.~P. Pryadko, ``{Quantum convolutional
  data-syndrome codes},'' in \emph{Proceedings of the IEEE 20th International
  Workshop on Signal Processing Advances in Wireless Communications (SPAWC)},
  Cannes, France, Jul. 2019.

\end{thebibliography}

\end{document}